\documentclass[11pt]{article}

\usepackage{a4wide, amsmath,amssymb,graphics, xspace}
\newtheorem{defin}{Definition}[section]
  
\newtheorem{theo}[defin]{Theorem}
  \newenvironment{theorem}{\begin{theo} \sl}{\end{theo}}
\newtheorem{lem}[defin]{Lemma}
  \newenvironment{lemma}{\begin{lem} \sl}{\end{lem}}
  
\newtheorem{propo}[defin]{Proposition}
  
\newtheorem{coro}[defin]{Corollary}
  
\newtheorem{obse}[defin]{Observation}
  
\newtheorem{rem}[defin]{Remark}
  
\newtheorem{myfact}[defin]{Fact}

\newenvironment{proof}{\emph{Proof.}}{\hfill $\Box$ \medskip\\}

\newcommand{\B}{\ensuremath{\mathcal{B}}}
\newcommand{\C}{\ensuremath{\mathcal{C}}}
\newcommand{\G}{\ensuremath{\mathcal{G}}}
\newcommand{\graph}{\G}
\newcommand{\T}{\mathcal{T}}
\newcommand{\M}{\mathcal{M}}

\newcommand{\Reals}{{\mathbb{R}}}            
\newcommand{\eps}{\varepsilon}               
\newcommand{\bd}{\partial}
\DeclareMathOperator{\polylog}{polylog}
\DeclareMathOperator{\myint}{Int}
\DeclareMathOperator{\myext}{Ext}

\DeclareMathOperator{\myweight}{{w}}
\newcommand{\weight}[1]{\myweight(#1)}

\newcommand{\etal}{{\emph{et~al.}\xspace}}

\def\dist{\mathbf{d}}
\def\disto{\dist_{w}}
\def\distg{\dist_{\graph}}
\def\distt{\dist_{\T}}

\def\distbo{\dist_{\B}}
\def\distso{\dist_{\sigma,w}}

\renewcommand{\leq}{\leqslant}
\renewcommand{\geq}{\geqslant}

\newcommand{\mytree}{T}                      
\newcommand{\tree}{\mytree}                   

\newcommand{\spath}{\sigma}
\newcommand{\myin}{\mathrm{in}}
\newcommand{\myout}{\mathrm{out}}
\newcommand{\Sin}{S_{\myin}}
\newcommand{\Sout}{S_{\myout}}

\newcommand{\BeginMyItemize}{\begin{itemize}\setlength{\itemsep}{-\parskip}}
\newcommand{\EndMyItemize}{\end{itemize}}
\newcommand{\myitemize}[1]{\BeginMyItemize #1 \EndMyItemize}

\newcommand{\BeginMyEnumerate}{\begin{enumerate}\setlength{\itemsep}{-\parskip}}
\newcommand{\EndMyEnumerate}{\end{enumerate}}
\newcommand{\myenumerate}[1]{\BeginMyEnumerate #1 \EndMyEnumerate}

\title{Geodesic Spanners for Points on a Polyhedral Terrain}
\author{Mohammad~Ali~Abam\thanks{Computer Engineering Department, Sharif University of Technology.
    Email: {\tt \{abam,mjrezaei\}@sharif.edu.}}
    \and
    Mark de Berg\thanks{Department of Computer Science, TU Eindhoven, the Netherlands.
    Email: {\tt mdberg@win.tue.nl}. MdB was supported by the Netherlands' Organisation
    for Scientific Research (NWO) under project no.~024.002.003.}
    \and
    Mohammad~Javad Rezaei Seraji$^*$
}
\date{}

\begin{document}
\maketitle

\begin{abstract}
Let $S$ be a set $S$ of $n$ points on a polyhedral terrain~$\T$ in $\Reals^3$, and let $\eps>0$ be a
fixed constant. We prove that $S$ admits a $(2+\eps)$-spanner with $O(n\log n)$ edges
with respect to the geodesic distance. This is the first spanner with constant spanning
ratio and a near-linear number of edges for points on a terrain. On our way to this result,
we prove that any set of $n$ weighted points in $\Reals^d$ admits an additively weighted
$(2+\eps)$-spanner with $O(n)$ edges; this improves the previously best known bound
on the spanning ratio (which was $5+\eps$), and almost matches the lower bound.
\end{abstract}

\newpage

\section{Introduction}

\paragraph{Background and motivation.}
When designing networks on a given set of nodes---whether they are road or railway networks,
or computer networks, or some other type of networks---there are often two conflicting
desiderata. On the one hand one would like to have fast connections between any pair
of nodes, and on the other hand one would like the network to be sparse.
This leads to the concept of \emph{spanners}, as defined next.

In an abstract setting, one is given a metric space $\M=(S, \dist_\M)$, where the
elements from $S$ are called \emph{points}---the points represent the nodes in the
network---and $\dist_\M$ is a metric on~$S$. A \emph{$t$-spanner} for $\M$,
for a given $t>1$, is an edge-weighted graph $\graph=(S,E)$ where the weight
of each edge $(p,q)\in E$ is equal to $\dist_\M(p,q)$ and the following condition
is satisfied: for all pairs $p,q\in S$ we have that $\distg(p,q)\leq t\cdot \dist_\M(p,q)$,
where $\distg(p,q)$ denotes the distance between $p$ and $q$ in $\graph$.
(The distance between $p$ and $q$ in $\graph$ is defined as the minimum weight of any path
connecting $p$ and $q$ in $\graph$.) In other words, the distance between any two points
in the spanner $\graph$ approximates their original distance in the metric space $\M$
up to a factor~$t$. The factor $t$ is called the \emph{spanning ratio} (or dilation,
or stretch factor) of~$\graph$. The question now becomes: can we construct a sparse
graph with small spanning ratio? Or, stated differently: given a desired spanning
ratio~$t$, how many edges do we need to obtain a $t$-spanner?

\paragraph{Previous work.}
As mentioned, the concept of spanners arises naturally in the design of efficient networks.
Spanners have also been used as a tool in solving a variety of other problems. It is not
surprising therefore that spanners have been studied extensively. Many papers on spanners
focus on obtaining spanners of small size, that is, with a small number of edges. This
is also the focus of our paper. However, other properties---spanners in which the total
weight of the edges is small, or spanners in which the maximum degree of the nodes is small---are
also of interest. Dynamic and kinetic spanners have been considered as well---see \cite{ab-ks-11, r-fdgs-12} for some recent results.

In the most general version, where we do not have any additional properties of the
underlying metric space, one can get a $(2k+1)$-spanner of size $O(n^{1+1/k})$,
for any integer~$k>0$ by the method given in \cite{addjj-osswg-93}.
In particular, in general metric spaces it is not known any way to obtain constant
spanning ratio with a spanner of size~$O(n \polylog n)$. For several special
types of metric spaces much better results can be obtained, however.

One important case is when $S$ is a set of $n$ points in $\Reals^d$ and
the \emph{Euclidean metric} is used. For any fixed $\eps>0$ one can then obtain a
$(1+\eps)$-spanner with $O(n)$ edges. More precisely, there is a $(1+\eps)$-spanner
with $O(n/\eps^{d-1})$ edges. See the book by Narasimhan and Smid~\cite{ns-gsn-07} for an
extensive discussion on geometric spanners. Another special case that received ample
attention~\cite{cgmz-ohrdm-05, cg-sdpss-06,gr-odsdms-08, hm-fcnld-05, t-bealdm-04}
are metric spaces of \emph{bounded doubling dimension}. (A metric space $\M=(S, \dist_\M)$ has
doubling dimension $d$ if any ball of radius $r$ in the space can be covered
by $2^d$ balls of radius~$r/2$.) Also for spaces of whose doubling dimension
is a constant---note that this is a generalization of Euclidean spaces---it is
possible to obtain, for any fixed~$\eps>0$, a $(1+\eps)$-spanner with $O(n)$ edges.

Another natural generalization to study is the case where the points in $S$ lie
on a polyhedral terrain $\T$ (or, more generally, on a surface) and the
\emph{geodesic distance} is used. A polyhedral terrain is the graph of a piecewise
linear function $f:D\rightarrow\Reals^3$, where $D$ is a convex polygonal region in the plane.
Polyhedral terrains, or {\sc tin}s, are often used in {\sc gis} to model mountainous landscapes.
The geodesic distance, $\distt(p,q)$, between two points $p,q\in\T$ is the length on the
shortest path on the terrain between $p$ and~$q$. We call a spanner for points on
a terrain with respect to the geodesic distance a \emph{geodesic spanner}.
At first sight it may seem that geodesic spanner are very similar to Euclidean spanners.
This is not the case: a crucial difference is that the metric space $(\T,\distt)$ does
not have bounded doubling dimension. Indeed, it is unknown whether any set of
points on a terrain admits a spanner with constant spanning ratio and of size~$O(n \polylog n)$.

There are two recent results that deal with what can be considered as special cases
of geodesic spanners.

First, consider the case where the terrain is completely flat
except for $n$ needle-like peaks, and the points in $S$ are located on the top of these peaks.
This leads to the concept of \emph{additively weighted spanners}, as studied by
Abam \etal ~\cite{abfgs-gswps-11}. Here one is given a set $S$ of points in $\Reals^2$
(or, more generally, in $\Reals^d$), where each $p\in S$ has a non-negative weight $\weight{p}$;
the weights model the heights of the peaks.
The additively weighted distance $\disto(p,q)$ between two points $p,q\in S$ is now defined as
\[
\disto(p,q) =
    \left\{ \begin{array}{ll}
              0 & \mbox{ if $p=q$,} \\
              \weight{p} + |pq| + \weight{q} & \mbox{ if $p \neq q$}
            \end{array}
    \right.
\]
where  $|.|$ denote the Euclidean distance. A $t$-spanner $\graph$ for the metric space
$(S, \disto)$ is called an \emph{additively weighted $t$-spanner}. Note that $(S, \disto)$
does not necessarily have bounded doubling dimension. (To see this,take a set
$S$ of $n$ points inside a unit disk in the plane, each having
unit weight.)
Nevertheless, Abam~\etal~\cite{abfgs-gswps-11} showed that there exists a $(5+\eps)$-spanner
$\graph$ with a linear number of edges for the metric space $(S,\disto)$.
They also proved that for any $\eps > 0$, there are weighted point sets $S$ such that any
$(2-\eps)$-spanner of $(S,\disto)$ has $\Omega(n^2)$ edges.

A second special case of spanners on a terrain is where the terrain is again completely
flat, except for a number of polygonal and plateaus at very high elevations,
and the points on $S$ are located on the flat part of the terrain. If the plateaus are
sufficiently high, then this terrain can be seen as a domain with polygonal holes.
Abam~\etal~\cite{aahz-gspipd-15} recently showed that for a set of $n$ points in
a polygonal domain with $h>0$ holes, there exists a $(5+\epsilon)$-spanner
of size~$O(n\sqrt{h}\log^2 n)$. When $h=0$, they obtain a $(\sqrt{10}+\epsilon)$-spanner
with $O(n \log^2 n)$ edges.

The main question is still open, however: is there a geodesic spanner with $O(n\polylog n)$ edges
and constant spanning ratio for any set of $n$ points on a terrain?

\paragraph{Our results.}
We answer the question above affirmatively by showing that, for any constant~$\eps>0$,
there exists a $(2+\eps)$-spanner with $O(n\log n)$ edges. Note that our result
not only generalizes the recent result of Abam~\etal~\cite{aahz-gspipd-15}, it also
improves both the spanning ratio and the size of the spanner. Also note that
the lower bound for additively weighted spanners implies that we cannot hope to
get spanning ratio $2-\eps$ with a subquadratic number of edges.
On the way to proving this result, we present a new algorithm to construct an
additively weighted spanner. This spanner has $O(n)$ edges, like the one of
Abam \etal ~\cite{abfgs-gswps-11} but its spanning ratio is $2+\eps$, an improvement
over the previously known bound of~$5+\eps$. Given the lower bound and the
fact that our spanner uses $O(n)$ edges, this is essentially optimal.
Our method to obtain a $(2+\eps)$-spanner on a terrain uses, besides the
additively weighted spanners, another tool that we believe is of independent interest:
we show that for any set of $n$ points on a terrain, there is a
\emph{balanced shortest-path separator}: a shortest path connecting two points
on $\bd\T$, or a triangle whose sides are shortest paths, that partition the
point set $S$ into two subsets of size at least $2n/9$.

\section{Additively weighted spanners for points in~$\Reals^d$}
In this section we present our improved spanner construction for additively weighted point sets.
We use the same global approach as Abam~\etal~\cite{abfgs-gswps-11}---we cluster the points
in a suitable way, then we construct a spanner on the cluster centers, and finally we connect
the points to the cluster centers---but the implementation of the various steps is different.
\medskip

Let~$S$ be the given weighted set of $n$ points in~$\Reals^d$ for which we want to
construct a spanner. We will partition $S$ into a number of clusters~$C_i\subseteq S$,
each with a designated center~$c_i\in C_i$,
such that the clusters have the following two properties.
Let $\C$ be the set of all cluster centers.
\myenumerate{
\item[(i)]
    The metric space $(\C,\disto)$ has doubling dimension~$O(d\log(1/\eps))$.
\item[(ii)]
    For any cluster~$C_i$ and any point $p\in C_i$, we have
    $\disto(p,c_i) \leq (2+\eps)\cdot\weight{p}$.
}
The following algorithm takes as input the weighted point set~$S$ and a parameter
$\eps>0$, and computes a clustering of~$S$ with these properties.
\myenumerate{
\item Sort the points of $S$ in non-decreasing order of their weight, with ties
    broken arbitrarily. Let $p_1,p_2,\ldots,p_n$ be the resulting sorted sequence.
\item Initialize the first cluster as $C_1 := \{p_1\}$,
      define $c_1 := p_1$ to be its center, and
      initialize the set of cluster centers as $\C := \{c_1\}$.
      Set $m := 1$ to be the current number of clusters.
\item Handle the points $p_2,\ldots,p_n$ in order, as follows.
      \myenumerate{
      \item Compute a center $c_j\in \C$ with $1 \leq j \leq m$ whose Euclidean distance
        to $p_i$ is minimum.
      \item If $|c_jp_i| \leq \eps \cdot \weight{p_i}$, then
             add $p_i$ to cluster~$C_j$.
            Otherwise, start a new cluster with $p_i$ as center:
            set $m := m+1$, set $C_m = \{p_i\}$ and $c_m := p_i$,
            and set $\C := \C \cup \{p_i\}$.
      }
\item Return the collection $\{C_1,\ldots,C_m\}$ of clusters, with $\C$ as cluster centers.
}
\begin{lemma}\label{le:cluster-computation}
The metric space $(\C,\disto)$ has doubling dimension~$O(d\log(1/\eps))$.
\end{lemma}
\begin{proof}
Consider a $\disto$-ball~$B(c_i,r)$ with radius~$r$ centered at a point $c_i\in \C$.
We must show that $B(c_i,r)$ can be covered by $2^{O(d\log(1/\eps))}$ balls of
radius~$r/2$. To this end, let $\C^*\subseteq B(c_i,r)$ be a maximal set of centers
such that $\disto(c_{j}, c_{k}) > r/2$ for every pair $c_{j}, c_{k} \in \C^*$.
Then the set of balls~$\{ B(c_j,r/2) : c_j \in C^*\}$ covers~$B(c_i,r)$.
Hence, suffices to prove that $|\C^*| = O(1/\eps^d)$.

Define $\C^*_1 := \{ c_j\in \C^* : \weight{c_j} \leq r/8\}$ and  $\C^*_2 := \C^* \setminus \C^*_1$.
Since the $\disto$-distance of any two points $c_j,c_k\in \C^*$ is at least $r/2$,
we have $|c_j c_k| \geq r/2 - \weight{c_j} - \weight{c_k}$. For $c_j,c_k\in \C^*_1$
this implies that $|c_j c_k| \geq r/4$. A simple packing argument shows that
we can only put $O(t^d)$ points whose mutual distances are at least $r/t$ into a ball
with radius $r$ in $\Reals^d$. We conclude that $|\C^*_1| = O(4^d)$.
To bound the size of $\C^*_2$ we use the fact that in our construction any two centers
$c_j,c_k\in \C$ satisfy
$|c_j c_k| > \eps\cdot \min(\weight{c_j},\weight{c_k})$. For $c_j,c_k\in \C^*_2$ this implies that
$|c_j c_k| > \eps \cdot (r/8)$. The packing argument now implies $|\C^*_2| = O( (8/\eps)^d)$.
Therefore, we have $|\C^*| = |\C^*_1|+ |\C^*_2| = O(4^d+1/\eps^d)= O(1/\eps^d)$.
\end{proof}
We can now compute a $(2+\eps)$-spanner~$\graph$ on $S$, for a given~$0\leq \eps \leq \sqrt{2}-1$, as follows.
\myenumerate{
\item[I.] Compute a clustering $\{C_1,\ldots,C_m\}$ and a set $\C$ of cluster centers,
    as described above.
\item[II.] Construct a $(1+\eps)$-spanner $\B$ on $\C$ using the method given
    by Gottlieb \etal~\cite{gr-odsdms-08} for computing spanners in spaces of bounded doubling dimension.
    The spanner produced by this method has the special property that the maximum degree
    in $\B$ is $O((2+1/\eps)^{O(\mathit{dim})})$, where $\mathit{dim}$ is the doubling dimension.
    We call $\B$ the \emph{backbone} of our spanner.
\item[III.] To obtain our final spanner~$\graph$, we connect each non-center point~$p$
    to the backbone: we connect~$p$ to the center $c_i$ of the cluster $C_i$ containing~$p$,
    and in addition we connect $p$ to all the neighbors of $c_i$ in $\B$.
}
\begin{theorem}\label{th:additive-spanner}
Let $S$ be a set of $n$ weighted points in $\Reals^d$, and let $\eps>0$
be a fixed constant. There exists a $(2+\eps)$-spanner
with $(2+1/\eps)^{O(d\log(1/\eps))}n$ edges for the metric space $(S, \disto)$.
\end{theorem}
\begin{proof}
The bound on the number of edges follows immediately from Lemma~\ref{le:cluster-computation}
together with the fact that
the maximum degree in the backbone~$\B$ is $O((2+1/\eps)^{O(\mathit{dim})})$,
where $\mathit{dim}$ is the doubling dimension.
%

It remains to prove the bound on the spanning ratio.
%
Let $\graph=(S,E)$ be the computed spanner. We must prove that
$\distg(p,q)\leq (2+\eps)\cdot \disto(p,q)$. If $(p,q)\in E$ this is obviously true.
If both $p$ and $q$ are centers then this is also true, since
the backbone $\B$ is a $(1+\eps)$-spanner on~$\C$. So now
consider the case where one or both of $p,q$ are non-center points.
Let $C_i$ and $C_j$ be the clusters containing~$p_i$ and $p_j$, respectively.
Note that our construction guarantees that
$|p c_i| \leq \eps \cdot \weight{p}$ and that
$\weight{p} \geq \weight{c_i}$;
two similar properties hold for $q$ and $c_j$.
(These properties are used in Inequalities~(\ref{eq1})--(\ref{eq4}) below.)
We consider two cases.
\begin{itemize}
\item The first case is that $p$ and $q$ belong to the same cluster, so $c_i=c_j$.
    We then have
    \begin{eqnarray}
    \distg(p,q)
        & \leq & \distg(p,c_i) + \distg(c_i,q)  \nonumber \\[1mm]
        &   =  & \big( \weight{p} + |pc_i| + \weight{c_i} \; \big) +
                 \big( \weight{c_i} +  |c_iq| + \weight{q} \; \big) \nonumber \\[1mm]
        & \leq & \big( \weight{p} + \eps \cdot \weight{p} + \weight{p} \; \big) +
                 \big( \weight{q} + \eps \cdot \weight{q} + \weight{q} \; \big) \label{eq1} \\[1mm]
        &  =  & (2+\eps) \cdot (\weight{p}+\weight{q}) \nonumber \\[1mm]
        & \leq & (2+\eps) \cdot (\weight{p}+|pq|+ \weight{q}) \nonumber \\[1mm]
        &  = & (2+\eps) \cdot \disto(p,q). \nonumber
    \end{eqnarray}
\item  The second case is that $p$ and $q$ belong to different clusters, so $c_i\neq c_j$.
    Since the backbone~$\B$ is a $(1+\eps)$-spanner on~$\C$, the shortest path in $\B$
    from $c_i$ to $c_j$ has length at most $(1+\eps)\cdot \disto(c_i,c_j)$.
    Define $c_s$ and $c_t$ to be the neighbors of $c_i$ and $c_j$ along this path, respectively.
    (If the path consists of two edges then $c_s = c_t$,
    and if it consists of a single edge then we define $c_s = c_t = c_j$.)
    Note that $(p,c_s)$ and $(c_t,q)$ are edges in~$\graph$. Hence,
    \[
    \begin{array}{lll}
    \distg(p,q) & \leq & \disto(p,c_s) + \distbo(c_s,c_t)+\disto(c_t,q) \\[1mm]
                 &  =   & \big( \weight{p} + |pc_s|+ \weight{c_s} \; \big)
                          + \distbo(c_s,c_t)
                          + \big( \weight{c_t} + |c_t q| + \weight{q} \; \big) \\[1mm]
                 & \leq & \weight{p} + |pc_i| + |c_i c_s| + \weight{c_s}
                          + \distbo(c_s,c_t) + \weight{c_t}+|c_t c_j| +|c_j q| + \weight{q}.
    \end{array}
    \]
    Moreover, because $(c_i,c_s,\ldots,c_t,c_j)$ is a shortest path in $\B$ we have
    \[
    \begin{array}{lll}
    \distbo(c_i,c_j) = \weight{c_i} + |c_i c_s| + \weight{c_s}
                            + \distbo(c_s,c_t)
                            + \weight{c_t} + |c_t c_j| + \weight{c_j}.
    \end{array}
    \]
    Since $\B$ is a $(1+\eps)$-spanner we thus get
    \[
    \begin{array}{lll}
    |c_i c_s| + \weight{c_s} + \distbo(c_s,c_t) + \weight{c_t} + |c_t c_j|
        &  = &  \distbo(c_i,c_j) - \weight{c_i} - \weight{c_j} \\[1mm]
        &  \leq &  (1+\eps) \cdot \disto(c_i,c_j) - \weight{c_i} - \weight{c_j}
    \end{array}
    \]
    It follows that
    \begin{eqnarray}
    \distg(p,q) & \leq &  \weight{p} + |pc_i| + (1+\eps) \cdot \disto(c_i,c_j)
                            - \weight{c_i} - \weight{c_j} + |c_j q| + \weight{q} \nonumber \\[1mm]
                 & \leq & \weight{p} + \eps \cdot \weight{p}
                            + (1+\eps)\cdot \disto(c_i,c_j)
                            - \weight{c_i} - \weight{c_j} + \eps \cdot \weight{q}
                            + \weight{q} \label{eq2} \\[1mm]
                 & \leq & \weight{p} + \eps \cdot \weight{p}
                            + (1+\eps)\cdot(\weight{c_i} + |c_ic_j| + \weight{c_j})
                            - \weight{c_i} - \weight{c_j} + \eps \cdot \weight{q}
                            + \weight{q} \nonumber \\[1mm]
                 &  =   & (1+\eps) \cdot (\weight{p} + \weight{q})
                           + (1+ \eps)\cdot|c_ic_j|
                           + \eps \cdot(\weight{c_i}+\weight{c_j}) \nonumber \\[1mm]
                 & \leq & (1+\eps) \cdot (\weight{p}+\weight{q})+ (1+ \eps) \cdot (|c_ip|+|pq|+|qc_j|) + \eps \cdot  (\weight{p}+\weight{q}) \label{eq3} \\[1mm]
                 & \leq & (1+2\eps) \cdot (\weight{p}+\weight{q})+ (1+ \eps) \cdot (\eps \cdot \weight{p}+|pq|+\eps \cdot \weight{q}) \label{eq4}\\[1mm]
                 &  =  & (1+3\eps+ \eps^2) \cdot (\weight{p}+\weight{q})+ (1+ \eps) \cdot |pq| \nonumber \\[1mm]
                 & \leq & (1+3\eps+ \eps^2) \cdot (\weight{p}+\weight{q}+|pq|) \nonumber \\[1mm]
                 &   =  & (1+3\eps+ \eps^2) \cdot \disto(p,q) \nonumber \\[1mm]
                 & \leq & (2+\eps) \cdot \disto(p,q) \nonumber
    \end{eqnarray}
    where the last inequality holds because we can assume without loss of generality
    that $\eps \leq \sqrt{2}-1$.
\end{itemize}
Thus in both cases we have $\distg(p,q) \leq (2+\eps) \cdot \disto(p,q)$.
\end{proof}
\section{Spanners for points on a polyhedral terrain}
Let $\T$ be a polyhedral terrain with $m$~vertices, and let $S$ be a set of $n$~points
on~$\T$. In this section we show that there is a $(2+\eps)$-spanner for $S$ with
respect to $\distt$, the geodesic distance on~$\T$.
Our global approach is divide-and-conquer: we partition $S$ into two subsets of roughly
equal size, compute spanners for these subsets recursively, and then generate a set
of edges to connect the points from the two subsets. For the latter step, it is important
that the two subsets are separated in a suitable way. In particular, we need to separate
the subsets by shortest paths (not necessarily between points in $S$). Next we define
the two types of separator that we allow more precisely, and we show that a suitable
separator always exists.
\medskip

The first type of separator is a shortest path $\spath(u,v)$ that connects two
points $u,v\in \bd\T$. Such a shortest path partitions $\Delta$ into two regions:
the closed region $\spath^+(u,v)$ consisting of all points $q\in\Delta$ that lie
to the right of the (directed) path~$\spath(u,v)$, and the open region~$\spath^-(u,v)$
consisting of all points to the left of~$\spath(u,v)$; see Fig.~\ref{fi:separator}(i).
Note that parts of~$\spath(u,v)$  may lie on $\bd\T$, so
$\myint(\spath^+(u,v))$ and, similarly, $\myint(\spath^-(u,v))$,
need not be connected where $\myint(.)$ denotes the interior.
\begin{figure}
\begin{center}
\includegraphics{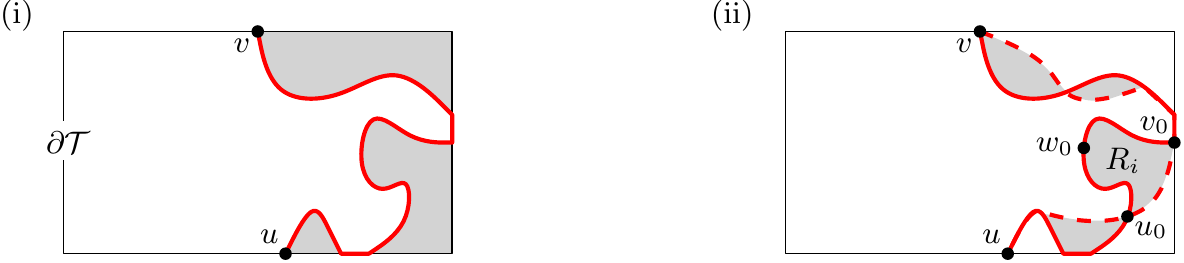}
\end{center}
\caption{(i) The shortest path $\spath(u,v)$ (in red) partitions $\T$ into two regions,
$\spath^+(u,v)$ (in grey) and $\spath^-(u,v)$ (in white). (ii) The paths
$\spath(u,v)$ and $\spath'(u,v)$ enclose a number of regions (in grey).
If $R_i$ contains more than $4n/9$ points, it is turned into an sp-triangle
by adding a point $w_0$ on one of the paths from $u_0$ to $v_0$.}
\label{fi:sp-separator}
\end{figure}

The second type of separators that we allow are defined as follows.
Consider three points $u,v,w\in\T$ with shortest paths $\spath(u,v)$, $\spath(v,w)$,
and $\spath(w,u)$ connecting them, and assume these paths are pairwise disjoint
except at shared endpoints. We call the closed region $\Delta$ bounded by
such a triple of paths a \emph{shortest-path triangle}, or \emph{sp-triangle} for short.
The paths $\spath(u,v)$, $\spath(v,w)$, and $\spath(w,u)$ are called the \emph{sides} of~$\Delta$.
A \emph{degenerate sp-triangle} is either a shortest path, or a path along~$\bd\T$.
In the following, when we talk about sp-triangles we also allow degenerate sp-triangles.

We call a separator of one of the two types defined above an \emph{sp-separator}.
The main tool in our spanner construction is the following theorem.
\begin{theorem} \label{th:balanced-triangle}
For any set of $n$ points on a polyhedral terrain~$\T$ there is a balanced sp-separator.
More precisely, there is either a shortest path $\spath(u,v)$ connecting two points
$u,v\in\bd\T$ such that $2n/9 \leq |\sigma^+(u,v)\cap S| \leq 2n/3$, or there is
an sp-triangle $\Delta$ such that $2n/9 \leq |\Delta\cap S| \leq 2n/3$.
\end{theorem}
To prove the theorem we first try to find a balanced sp-separator of
the first type. If this fails we argue that a suitable sp-triangle exists.
\medskip

Let $u\not\in S$ be an arbitrary point on $\bd\T$. Now move a point~$v$ around $\bd\T$,
starting at $u$ and traversing $\bd\T$ counterclockwise, until $v$
reaches~$u$ again. As we continuously move $v$ along~$\bd\T$, the shortest
path~$\spath(u,v)$ also changes continuously, except at certain \emph{breakpoints}.
More precisely, we can partition $\bd\T$ into finitely many pieces---the breakpoints
are the endpoints of these pieces---such that as $v$ moves along one such a boundary piece,
we can deform~$\spath(u,v)$ continuously. Initially, when $v$ is still
infinitesimally close to $u$, we have $\spath^+(u,v)\cap S = \emptyset$;
at the very end, when $v$ approaches $u$ again, we have $\spath^+(u,v)\cap S = S$.

If at some point during the traversal $v$ reaches a position such that
$2n/9\leq |\spath^+(u,v) \cap S| \leq 2n/3$ then we have found our balanced
sp-separator. Otherwise there is a breakpoint~$v^*$ at which $\spath(u,v)$ jumps
over more than $2n/3-2n/9=4n/9$ points from~$S$. In this case there are two shortest
paths $\spath(u,v^*)$ and $\spath'(u,v^*)$ such that the (open) region $R$
enclosed by $\spath(u,v^*)$ and $\spath'(u,v^*)$ contains more than $4n/9$ points from~$S$.

Note that $R$ may consist of more than one connected components. If all of them contain
at most $4n/9$ points from $S$, then we can obtain a balanced separator by only
jumping over a subset of the components. Otherwise there is a single component, $R_i$,
that contains more than $4n/9$ points. Let $u_0$ and $v_0$ be the two points on
$\bd R_i$ where $\spath(u,v^*)$ and $\spath'(u,v^*)$ meet. Let $w_0$ be an arbitrary
point on $\bd R_i$ that is distinct from $u_0$ and $v_0$; see Fig.~\ref{fi:sp-separator}(ii).
Then the triple $u_0,v_0,w_0$, together with $\bd R_i$, defines an sp-triangle
$\Delta_0$ containing at least $4n/9$ points from~$S$.
Next we show how to construct a sequence of sp-triangles
$\Delta_0 \supset \Delta_1 \supset \cdots \supset \Delta_k$ such that
$2n/9 \leq |\Delta_{k}\cap S| \leq 2n/3$.
We will maintain the invariant that $|\Delta_i\cap S|\geq 2n/9$.
Note that this is indeed satisfied for~$\Delta_0$.
In the following, we denote the vertices of the sp-triangle~$\Delta_i$ by $u_i,v_i,w_i$,
its interior by $\myint(\Delta_i)$, and its boundary by $\bd\Delta_i$.
\medskip

Suppose we have
constructed $\Delta_i$. If $|\Delta_i\cap S|\leq 2n/3$ then $\Delta_i$ is the final
triangle in our construction and we are done.
If $\Delta_i$ is a degenerate sp-triangle, then we can immediately
find an sp-triangle $\Delta_{i+1}\subset \Delta_i$ with the required
properties: we just take a subpath containing $2n/9$ points.
It remains to handle the case where $\Delta_i$ is a non-degenerate sp-triangle
containing more than $2n/3$ points. Note that if $\Delta_i$ has a side
containing at least $2n/9$ points from $S$ we can again finish the
construction by taking a suitable subpath of this side as our next (and final) sp-triangle.
Hence, we can assume that each side contains less than $2n/9$ points,
which implies there is at least one point---actually, at least four
points---in the interior of~$\Delta_i$.
Next we show how to construct an sp-triangle~$\Delta_{i+1}\subset \Delta_i$ containing
at least~$2n/9$ points such that either
$|\Delta_{i+1}\cap S| < |\Delta_{i}\cap S|$ or
$|\myint(\Delta_{i+1})\cap S| < |\myint(\Delta_{i})\cap S|$.
Note that the condition that $|\Delta_{i+1}\cap S| < |\Delta_{i}\cap S|$ or
$|\myint(\Delta_{i+1})\cap S| < |\myint(\Delta_{i})\cap S|$ implies that our process terminates.
Indeed, when $|\myint(\Delta_{i+1})\cap S|=0$ and $\Delta_{i+1}$ contains
more than $2n/3$ points, then we have a side with more than~$2n/9$ points
and so we can finish the construction as described above.


To simplify the notation, we will drop the subscript~$i$ from now on. Thus we
are given a non-degenerate sp-triangle $\Delta$ with corners $u,v,w$ that contains
more than $2n/3$ points from $S$ and has at least one point in its interior.
For a point $z\in \Delta$ we call a path from $z$ to one of the corners of
$\Delta$ a \emph{good path} if it is a shortest path that stays within~$\Delta$
(although not necessarily in its interior).
Define $Z_u\subseteq \Delta$ to be the set of all points $z\in\Delta$ such that
there is a good path $\spath(z,u)$ to the corner~$u$.
The region~$Z_u$ is simply connected and its boundary
consists of the sides $\spath(u,v)$ and $\spath(u,w)$, and a
curve $B_u$ from $v$ to $w$---see Fig.~\ref{fi:def-Bu}(i). Note that~$B_u$
may overlap partially (or fully) with $\spath(u,v)$ and/or $\spath(u,w)$.
\begin{figure}
\begin{center}
\includegraphics{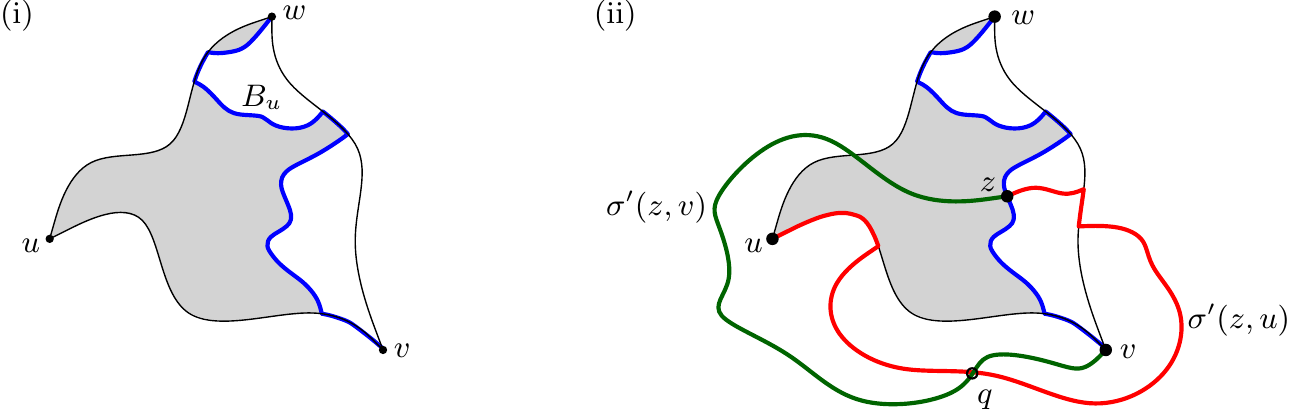}
\end{center}
\caption{(i) The blue curve is $B_u$. The interior of the region~$Z_u$ is indicated
in grey; the curve $B_u$ is also part of~$Z_u$. (ii) The paths $\sigma'(z,u)$ and $\sigma'(z,v)$
intersect.}
\label{fi:def-Bu}
\end{figure}
\begin{lemma}\label{le:stay-inside}
For any point $z\in B_u$, there are good paths $\spath(z,u)$, $\spath(z,v)$,
and $\spath(z,w)$ to the three corners of $\Delta$.
\end{lemma}
\begin{proof}
By definition of $B_u$, there is a good path from $z$ to~$u$.
If $z$ lies on $\spath(v,w)$, the side of $\Delta$ opposite~$u$, then we also
have good paths to $v$ and~$w$. Now assume $z\not\in \spath(v,w)$.

Since $z\in \bd Z_u$ and $z\not\in \spath(v,w)$ we have not only a good
path from $z$ to~$u$, but we also have a shortest path $\spath'(z,u)$ that
does not stay inside $\Delta$ and that cannot be shortcut (while maintaining
its length) to do so. Note that as soon as $\spath'(z,u)$ exits $\Delta$ through one of the
sides $\spath(u,v)$ or $\spath(u,w)$ it could also follow that side to~$u$.
Hence, we can assume $\spath'(z,u)$ is as follows: it exits $\Delta$ through
$\spath(v,w)$ (possibly after following $\spath(v,w)$ for a while),
then it moves through $\myext(\Delta)$ until it hits one of the sides
incident to~$u$, say $\spath(u,w)$, which it then follows to~$u$.
See~Fig.~\ref{fi:def-Bu}(ii). (The portion along $\spath(u,w)$ may be empty.)
Note that within $\myext(\Delta)$ the path $\spath(z,u)$ separates~$u$
either from $v$ or from $w$. Assume without loss of generality that the former is the
case, as in Fig.~\ref{fi:def-Bu}(ii). We now argue the existence of good paths
$\spath(z,v)$ and $\spath(z,w)$.

First consider a shortest path $\spath'(z,v)$. If $\spath'(z,v)$ already stays inside $\Delta$
we have a good path to~$v$. Otherwise it must go through~$\myext(\Delta)$ and, hence,
cross $\spath'(z,u)$ at a point~$q$. Note that the distances from $z$ to $q$ along
$\spath'(z,u)$ and along $\spath'(z,v)$ must be equal. But then the path from $z$
to $v$ that follows $\spath'(z,u)$ until it hits the side $\spath(v,w)$ and then
follows that side to~$v$ cannot be longer than~$\spath'(z,v)$.
Hence, we have a good path from $z$ to $v$.

Now consider a shortest path $\spath'(z,w)$. If it doesn't already stay inside~$\Delta$,
it exists $\Delta$ through the side~$\spath(u,v)$ and it crosses~$\spath'(z,u)$.
In the latter case we can use the same argument as above, and find a good path to~$w$.
\end{proof}
Now imagine moving a point $z$ from $w$ to $v$ along~$B_u$. By the previous lemma,
at any point we have good paths $\spath(z,u)$ and $\spath(z,v)$.
Now consider a \emph{shortest-path tree}\footnote{The two shortest paths
$\spath(z,u)$ and $\spath(z,v)$ may not form a tree because they meet more than
once, but by re-routing we can always get rid of this situation.}
$\tree_z$ with $z$ as the root and $u$ and $v$ as leaves
that consists of good paths $\spath(z,u)$ and $\spath(z,v)$ to~$u$ and~$v$---see Fig.~\ref{fi:sp-tree}(i).
\begin{figure}
\begin{center}
\includegraphics{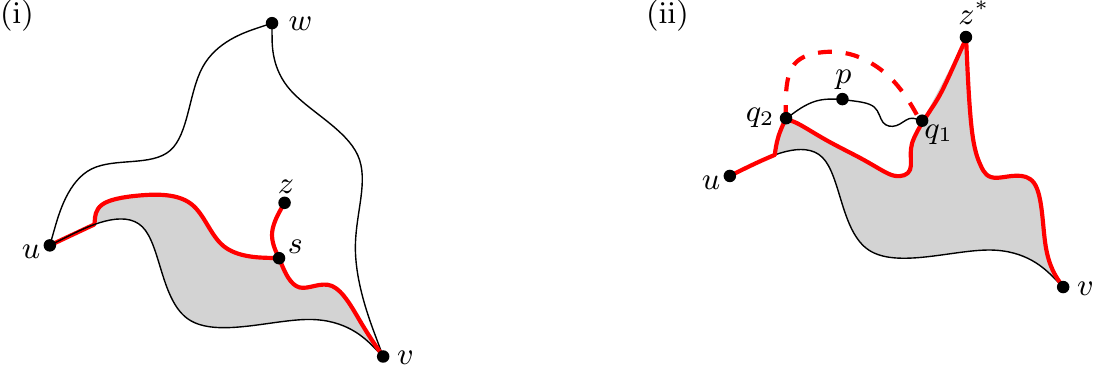}
\end{center}
\caption{(i) A good-path tree (in red) with root $z$ and leaves $u$ and $v$. The interior of
$\Delta_z$ is indicated in grey. (ii) A subpath of the shortest path~$\spath(z,u)$ jumps at~$z$
from the dashed subpath to the solid subpath from $q_1$ to $q_2$. Region~$R$ is split into
two sp-triangles by connecting $p$ to $q_1$ and $q_2$ by shortest paths.}
\label{fi:sp-tree}
\end{figure}
When $z=w$ we take $\spath(z,u)$ and $\spath(z,v)$ to be the sides
$\spath(w,u)$ and $\spath(w,v)$ of $\Delta$. When $z=v$
we take $\spath(z,u)$ to be the side $\spath(v,u)$; the path $\spath(z,v)$ is then empty.
Let $\Delta_z$ denote the sp-triangle with $v$ as one of its corners that is
bounded by (part of) $\tree_z$ and (part of) the side $\spath(u,v)$.

As we continuously move $z$ along~$B_u$, the good-path tree~$\tree_z$ also changes continuously,
except at certain breakpoints.  Initially, when $z=w$, we have
$\tree_z=\spath(w,u)\cup \spath(w,v)$ and so $\Delta_z =\Delta$.
At the end, when $z=v$, we have $\tree_z=\spath(u,v)$ and so
$\myint(\Delta_z)$ is empty. (Thus $\Delta$ is a degenerate sp-triangle.)
Now consider the first moment when either $|\Delta_z \cap S|$ decreases
or $|\myint(\Delta_z) \cap S|$ decreases. Let $z^*$ be the point at which this happens.
We have two cases.
\myitemize{
\item If $|\Delta_{z^*}\cap S|  \geq 2n/9$ then we can take $\Delta_{i+1} := \Delta_{z^*}$.
\item Otherwise, the number of points we just lost is more than $2n/3 - 2n/9=4n/9$.
      This can happen when one of the two good paths forming $\tree_z$ jumps.
      In this case~$z^*$ must be a breakpoint, and there are
      two different shortest paths from $z^*$ to $u$, or two different shortest paths
      from $z^*$ to~$v$ (or both).
      Assume we have two shortest paths from $z^*$ to~$u$. The points that we lose as $z$ moves
      over the breakpoint are located in between these two paths. The area
      between the two paths may consist of various regions. If one of these regions
      contains at least one and at most $4n/9$ points, then we can find a suitable sp-triangle $\Delta_{i+1}$ by only jumping over this region. Otherwise we have a region~$R$
      that contains more than~$4n/9$ points. This region is bounded by two
      shortest paths that meet at shared endpoints, $q_1$ and $q_2$. We then take
      any point $p\in \myint(R)\cap S$, and connect $p$ by shortest paths that
      stay in $R$ to $q_1$ and $q_2$; see Fig.~\ref{fi:sp-tree}(ii).
      This partitions $R$ into two sp-triangles. At least one of them
      contains more than $2n/9$ points. We take this sp-triangle to be~$\Delta_{i+1}$.
      Note that $\Delta_{i+1}$ contains less points in its interior than
      $\Delta$, since $p\not\in\myint(\Delta_{i+1})$.
}
In both cases we find an sp-triangle with the required properties, thus
finishing the proof of Theorem~\ref{th:balanced-triangle}. Next we describe
how to use this theorem to compute a spanner for a set $S$ of points
on a terrain.

\paragraph{The spanner construction.}
Next we show how to compute a spanner~$\graph(S)=(S,E_S)$ for
a set $S$ of $n$ points on a terrain~$\T$. We first describe how to obtain a
$(6+\eps)$-spanner, then we show how to improve the construction
to reduce the spanning ratio to~$(2+\eps)$.
\medskip

Let $\eps>0$ be a given constant.
If $|S|\leq 3$ we connect all pairs of points in $S$, that is,
$\graph(S)$ is the complete graph. $|S|> 3$ we proceed as follows.
\begin{enumerate}
\item Take a balanced sp-separator, as in Theorem~\ref{th:balanced-triangle}.
      If the separator is a shortest path $\spath(u,v)$ connecting points
      $u,v\in\bd\T$, then define $\Sin := \spath^+(u,v) \cap S$; if the separator
      is an sp-triangle $\Delta$ then define $\Sin := \Delta\cap S$.
      If the first case applies, we will from now on call the shortest path $\spath(u,v)$
      the \emph{side} of the separator. Thus a separator has one or three sides.
      Note that a side is always a shortest path on $\T$.
      Set $\Sout := S\setminus \Sin$. 
\item Process each side $\spath$ of the separator as follows. \label{step:merge-edges}
      \myenumerate{
      \item[(i)] For each $p\in S$, let $p_{\sigma} \in \sigma$ be a point whose geodesic
            distance $\distt(p,p_\sigma)$ to $p$ is minimum. Assign each point
            $p_{\sigma}$ a weight $\weight{p_{\sigma}} := \distt(p,p_\sigma)$ and
            define $S_{\sigma} := \{ p_{\sigma} : p\in S \}$ to be the resulting weighted
            (multi-)set.
      \item[(ii)] We view $\sigma$ as a 1-dimensional Euclidean space, and $S_\sigma$ as a
            weighted point set in this space. (Note that because $\sigma$ is a shortest
            path on $\T$, distances $d_{\sigma}$ in the 1-dimensional space~$\sigma$
            are the same as distances $\distt$ on~$\T$.) We now construct an additively weighted
            $(2+\eps_1)$-spanner $\graph_{\sigma}$ for $S_{\sigma}$, where $\eps_1=\eps/3$,
            using the method from Theorem~\ref{th:additive-spanner}.
      \item[(iii)] For each edge $(p_\sigma,q_\sigma)$ in the spanner $\graph_{\sigma}$,
            we add $(p,q)$ to $E(S)$.
      }
\item Recursively compute spanners $\graph(\Sin)=(\Sin,E_{\myin})$ and
      $\graph(\Sout)=(\Sout,E_{\myout})$, and add the edge sets $E_{\myin}$ and
      $E_{\myout}$ to $E$.
\end{enumerate}
\begin{lemma}\label{le:ratio6}
The construction above gives a $(6+\eps)$-spanner with respect to the geodesic distance.
The spanner has $O(c_{\eps}n\log n)$ edges, where $c_{\eps}$ is a constant depending on~$\eps$.
\end{lemma}
\begin{proof}
By Theorem~\ref{th:additive-spanner} the number of edges we add to the spanner in
Step~\ref{step:merge-edges} is $O(n)$.
Hence, if $A(n)$ denotes the total number of edges we have in our spanner on $n$
points, then we have $A(n) = O(c_\eps n) + A(n_1) + A(n_2)$, where $n_1+n_2=n$ and
$n_1,n_2\leq 7n/9$ and $c_\eps = (2+1/\eps)^{O(\log(1/\eps))}n$. Hence, $A(n) = O(c_\eps n\log n)$ as claimed.

Let $p,q$ be two arbitrary points in $S$. If both points are in $\Sin$ or both points are
in $\Sout$ then we have a $(2+\eps_1)$-path between them by induction. So assume
$p\in \Sin$ and $q\in \Sout$. Let $\spath(p,q)$ be a shortest path on $\T$ from~$p$ to $q$.
Let $\sigma$ be a side of $\Delta$ intersected by $\spath(p,q)$.
Consider the points $p_{\sigma}$ and $q_{\sigma}$. Then there is a path
$\pi$ in $\graph_{\sigma}$ of length at most $(2+\eps_1)\cdot \distso(p_\sigma,q_\sigma)$,
where $\distso$ denotes the additively weighted distance in the 1-dimensional space~$\sigma$. 
The same path, with each
point $x_{\sigma}$ replaced by its original $x\in S$, also exists in our
spanner~$\graph$. Note that for any two points $x,y\in S$ we have
\[
\begin{array}{lll}
\distt(x,y) & \leq & \distt(x,x_{\sigma}) + \distt(x_\sigma,y_\sigma) + \distt(y_{\sigma},y) \\[1mm]
            &   =  & \weight{x_\sigma} + d_{\sigma}(x_\sigma,y_\sigma) + \weight{y_\sigma} \\[1mm]
            &   =  & \distso(x,y).
\end{array}
\]
Let $r$ be a point in $\sigma\cap \spath(p,q)$. Then $\distt(p,q) = \distt(p,r)+\distt(r,q)$.
We also have $\distt(p,p_\sigma) \leq \distt(p,r)$ and $\distt(q,q_\sigma) \leq \distt(q,r)$
by definition of $p_\sigma$ and $q_\sigma$. Hence,
\begin{eqnarray}
\distg(p,q) & \leq & \mbox{length}(\pi) \nonumber \\[1mm]
            & \leq & (2+\eps_1) \cdot \distso(p_\sigma,q_\sigma) \label{eq:2to6} \\[1mm]
            & \leq & (2+\eps_1) \cdot \left( \weight{p_\sigma} +   d_{\sigma}(p_\sigma,q_\sigma) + \weight{q_\sigma} \right) \nonumber \\[1mm]
            &   =  & (2+\eps_1) \cdot
            \left( \distt(p,p_\sigma) + \distt(p_\sigma,q_\sigma) + \distt(q,q_\sigma) \right) \nonumber \\[1mm]
            & \leq & (2+\eps_1) \cdot
            \left( \distt(p,r) + \distt(p_\sigma,q_\sigma) + \distt(q,r) \right) \nonumber \\[1mm]
            &   =  & (2+\eps_1) \cdot
            \left( \distt(p,q) + \distt(p_\sigma,q_\sigma) \right) \nonumber
\end{eqnarray}
Moreover,
\[
\begin{array}{lll}
\distt(p_\sigma,q_\sigma) & \leq & \distt(p_\sigma,p) + \distt(p,q) + \distt(q,q_\sigma) \\[1mm]
                          & \leq & \distt(r,p) + \distt(p,q) + \distt(q,r) \\[1mm]
                          &  =   & 2 \; \distt(p,q).
\end{array}
\]
Thus $\distg(p,q) \leq (2+\eps_1) \cdot 3 \; \distt(p,q) = (6+\eps)\cdot \distt(p,q)$.
\end{proof}
We now refine our construction to reduce the spanning ratio to~$2+\eps$.
The idea behind the improvement is as follows. As follows from the proof
of Lemma~\ref{le:ratio6}, we would get a $(2+\eps_1)$-spanner if we had
$\distso(p_\sigma,q_\sigma) = \distt(p,q)$. In the above construction, however,
we have $\distso(p_\sigma,q_\sigma) \leq 3\;\distt(p,q)$ giving a $(6+\eps)$-spanner.
We can obtain $\distso(p_\sigma,q_\sigma) \leq (1+\eps_2) \cdot \distt(p,q)$,
for a suitable $\eps_2=O(\eps)$,
by modifying Step~\ref{step:merge-edges}  as follows.

In Step~\ref{step:merge-edges}(i) we take for each~$p\in S$ not a single point $p_\sigma\in\sigma$
but a collection $S(p,\sigma)$ defined as follows. As before, let $p_\sigma$ be a point
on $\sigma$ that is closest to~$p$. Let $\sigma(p) \subseteq \sigma$ be the set of points
on $\sigma$ whose distance to $p_\sigma$ is at most $(1+2/\eps_2)\cdot\distt(p,p_\sigma)$, that is,
\[
\sigma(p) := \{ x\in \sigma : \distso(p_\sigma,x) \leq (1+2/\eps_2)\cdot\distt(p,p_\sigma) \}.
\]
We partition $\sigma(p)$ into $O(1/\eps_2^2)$ pieces $\sigma_i(p)$, each of length
at most~$\eps_2 \cdot \distt(p,p_\sigma)$. For each such piece $\sigma_i(p)$, let
$p_\sigma^{(i)}$ a point on $\sigma_i(p)$ that is closest to~$p$. We now take
$S(p,\sigma)$ to be the set of all such points~$p_\sigma^{(i)}$, where we
set $\weight{p_\sigma^{(i)}} := \distt(p,p_\sigma^{(i)})$.
Note that $|S(p,\sigma)| = O(1/\eps_2^2)$.

In Step~\ref{step:merge-edges}(ii) we now compute a $(2+\eps_1)$ spanner $\graph_\sigma$
on the set $\bigcup_{p\in S} S(p,\sigma)$. In Step~\ref{step:merge-edges}(iii) we then add
for each edge $(p_\sigma^{(i)},q_\sigma^{(j)})$ in $\graph_\sigma$ the edge~$(p,q)$
to~$E(S)$. This leads to the following result.
\begin{theorem}\label{th:terrain-spanner}
Let $S$ be a set of $n$ points on a polyhedral terrain~$\T$ in $\Reals^3$, and let $\eps>0$
be a fixed constant. Then there exists  a $(2+\eps)$-spanner with
$O(c_{\eps}n\log n)$ edges with respect to the geodesic distance, where $c_{\eps}$
is a constant depending on~$\eps$.
\end{theorem}
\begin{proof}
To prove the bound on the spanning ratio, we observe that as compared to our previous
construction the number of points for which we compute an additively weighted $(2+\eps_1)$-spanner
in Step~\ref{step:merge-edges}(ii) has increased from $O(n)$ to $(n/\eps_2^2)$.
Hence, if we set $\eps_2 = O(\eps)$ the overall number of edges increases by a factor $O(1/\eps^2)$.

It remains to prove the bound on the spanning ratio of our spanner~$\graph$.
To this end, consider two points $p\in \Sin$ and $q\in\Sout$. As in the
proof of Lemma~\ref{le:ratio6}, let $r$ be a point where the shortest path
$\sigma(p,q)$ crosses~$\sigma$.  If $r \not\in \sigma(p)$, set $p' := p_\sigma$. Otherwise, set $p'$ to be the closest point in $S(p,\sigma)$ to $r$.
Similarly define $q'$ for point $q$. Note that
\[
\begin{array}{lll}
\distso(p',q') & = & \weight{p'} +   d_{\sigma}(p',q') + \weight{q'} \nonumber \\[1mm]
                  & = &   \distt(p,p') +   d_{\sigma}(p',q') + \distt(q,q')  \nonumber  \\[1mm]
                  & \leq & \distt(p,p') +   d_{\sigma}(p',r) + d_{\sigma}(r,q')+ \distt(q,q').  \nonumber  \\[1mm]
\end{array}
\]
We next prove that $ \distt(p,p') +   d_{\sigma}(p',r) \leq (1+\eps_2)\cdot  \distt(p,r)$.
We have two cases:
\begin{itemize}
\item \emph{Case A:} $r \not\in \sigma(p)$.
In this case $p'= p_\sigma$ and $d_\sigma(p',r) > (1+2/\eps_2)\cdot \distt(p,p_\sigma)$.
Hence,
\[
\begin{array}{lll}
 \distt(p,p') +   d_{\sigma}(p',r) &\leq& \distt(p,p_\sigma) + (\distt(p,p_\sigma)+\distt(p,r)) \\[1mm]
                                               &\leq& 2\cdot \distt(p,p_\sigma) +  \distt(p,r) \\[1mm]
                                               &\leq& 2 \cdot (\eps_2/2) \cdot (d_\sigma(p',r)-\distt(p,p_\sigma)) +  \distt(p,r) \\[1mm]
                                               &\leq& \eps_2 \cdot (\distt(p_\sigma,p)+\distt(p,r)-\distt(p,p_\sigma)) +  \distt(p,r)  \\[1mm]
                                              & = & (1+\eps_2)\cdot \distt(p,r).
\end{array}
\]

\item \emph{Case B:} $r \in \sigma(p)$.
Now we have $\dist_\sigma(p',r) \leq \eps_2 \cdot \distt(p,p_\sigma)$, and so
\[
\begin{array}{lll}
 \distt(p,p') +   d_{\sigma}(p',r) &\leq& \distt(p,r) + \eps_2 \cdot \distt(p,p_\sigma) \\[1mm]
                                               &\leq& \distt(p,r) + \eps_2 \cdot \distt(p,r) \\[1mm]
                                               &=& (1+\eps_2)\cdot \distt(p,r).
\end{array}
\]

\end{itemize}
So in both cases we have $\distt(p,p') +   d_{\sigma}(p',r) \leq (1+\eps_2)\cdot  \distt(p,r)$.
In a similar way we can prove that $\distt(q,q') +   d_{\sigma}(q',r) \leq (1+\eps_2)\cdot  \distt(q,r)$. Hence, $\distso(p',q') \leq (1+\eps_2) \cdot  \distt(p,q)$.
Combing this with Inequality~(\ref{eq:2to6}) from the proof of Lemma~\ref{le:ratio6},
which now holds with $p_\sigma$ replaced by $p'$, we obtain
$\distg(p,q) \leq (2+\eps_1) \cdot (1+\eps_2) \cdot  \distt(p,q)$.
Picking $\eps_1=\eps_2=\eps/4$ now gives us the desired spanning ratio
(assuming without loss of generality that $\eps\leq 1$).
\end{proof}

\section{Concluding remarks}
We have shown that any set of $n$ points on a polyhedral terrain $\T$ admits a geodesic spanner
of spanning ratio $2+\eps$ and with $O(n\log n)$ edges. This is the first geodesic spanner for
points on a terrain. In fact, our method works in a more general setting than for polyhedral
terrains: it suffices to have a piecewise-linear surface that is a topological disk.
(In fact, our method also works for smooth surfaces, under certain mild conditions
that make shortest paths be well behaved.)
In the current paper we have focused on the proving the existence
of sparse geodesic spanner, leaving the efficient computation of such spanners to future research.


\end{document}